\documentclass[12pt]{article}

\IfFileExists{sariel_computer.sty}{\def\sarielComp{1}}{}
\ifx\sarielComp\undefined%
\newcommand{\SarielComp}[1]{}
\newcommand{\NotSarielComp}[1]{#1}%
\else
\newcommand{\SarielComp}[1]{#1}%
\newcommand{\NotSarielComp}[1]{}%
\fi
\newcommand{\IfPrinterVer}[2]{#2}%

\usepackage[cm]{fullpage}%
\usepackage{amsmath}%
\usepackage{amssymb}%
\usepackage{xcolor}%
\usepackage{euscript}%
\usepackage{caption}%

\SarielComp{\usepackage{sariel_colors}}%

\usepackage[amsmath,thmmarks]{ntheorem}%
\theoremseparator{.}%

\usepackage{titlesec}%
\titlelabel{\thetitle. }%
\usepackage{xcolor}%
\usepackage{mleftright}%
\usepackage{xspace}%
\usepackage{hyperref}%
\usepackage{graphicx}%

\IfPrinterVer{%
   \usepackage{hyperref}%
}{%
   \usepackage{hyperref}%
   \hypersetup{%
      breaklinks,%
      ocgcolorlinks, colorlinks=true,%
      urlcolor=[rgb]{0.25,0.0,0.0},%
      linkcolor=[rgb]{0.5,0.0,0.0},%
      citecolor=[rgb]{0,0.2,0.445},%
      filecolor=[rgb]{0,0,0.4},
      anchorcolor=[rgb]={0.0,0.1,0.2}%
   }
}

\theoremseparator{.}%

\theoremstyle{plain}%
\newtheorem{theorem}{Theorem}[section]

\newtheorem{lemma}[theorem]{Lemma}

\newtheorem{claim}[theorem]{Claim}%

\newtheorem{observation}[theorem]{Observation}

\theoremstyle{plain}%
\theoremheaderfont{\sf} \theorembodyfont{\upshape}%
\newtheorem*{remark:unnumbered}[theorem]{Remark}%

\newtheorem{defn}[theorem]{Definition}

\newcommand{\myqedsymbol}{\rule{2mm}{2mm}}

\theoremheaderfont{\em}%
\theorembodyfont{\upshape}%
\theoremstyle{nonumberplain}%
\theoremseparator{}%
\theoremsymbol{\myqedsymbol}%
\newtheorem{proof}{Proof:}%

\definecolor{blue25emph}{rgb}{0, 0, 11}
\providecommand{\emphic}[2]{%
   \textcolor{blue25emph}{%
      \textbf{\emph{#1}}}%
   \index{#2}}

\providecommand{\emphi}[1]{\emphic{#1}{#1}}

\definecolor{almostblack}{rgb}{0, 0, 0.3}
\providecommand{\emphw}[1]{{\textcolor{almostblack}{\emph{#1}}}}%

\DeclareFontFamily{U}{BOONDOX-calo}{\skewchar\font=45 }
\DeclareFontShape{U}{BOONDOX-calo}{m}{n}{
  <-> s*[1.05] BOONDOX-r-calo}{}
\DeclareFontShape{U}{BOONDOX-calo}{b}{n}{
  <-> s*[1.05] BOONDOX-b-calo}{}
\DeclareMathAlphabet{\mathcalb}{U}{BOONDOX-calo}{m}{n}
\SetMathAlphabet{\mathcalb}{bold}{U}{BOONDOX-calo}{b}{n}
\DeclareMathAlphabet{\mathbcalb}{U}{BOONDOX-calo}{b}{n}

\newcommand{\HLink}[2]{\hyperref[#2]{#1~\ref*{#2}}}
\newcommand{\HLinkSuffix}[3]{\hyperref[#2]{#1\ref*{#2}{#3}}}

\newcommand{\figlab}[1]{\label{fig:#1}}
\newcommand{\figref}[1]{\HLink{Figure}{fig:#1}}

\newcommand{\thmlab}[1]{{\label{theo:#1}}}
\newcommand{\thmref}[1]{\HLink{Theorem}{theo:#1}}

\newcommand{\seclab}[1]{\label{sec:#1}}
\newcommand{\secref}[1]{\HLink{Section}{sec:#1}}

\newcommand{\apndlab}[1]{\label{apnd:#1}}
\newcommand{\apndref}[1]{\HLink{Appendix}{apnd:#1}}

\newcommand{\obslab}[1]{\label{observation:#1}}
\newcommand{\obsref}[1]{\HLink{Observation}{observation:#1}}

\providecommand{\deflab}[1]{\label{def:#1}}
\newcommand{\defref}[1]{\HLink{Definition}{def:#1}}

\newcommand{\clmlab}[1]{\label{claim:#1}}
\newcommand{\clmref}[1]{\HLink{Claim}{claim:#1}}

\newcommand{\lemlab}[1]{\label{lemma:#1}}
\newcommand{\lemref}[1]{\HLink{Lemma}{lemma:#1}}%

\providecommand{\eqlab}[1]{}%
\renewcommand{\eqlab}[1]{\label{equation:#1}}
\newcommand{\Eqref}[1]{\HLinkSuffix{Eq.~(}{equation:#1}{)}}

\newcommand{\remove}[1]{}%
\newcommand{\Set}[2]{\left\{ #1 \;\middle\vert\; #2 \right\}}
\newcommand{\pth}[2][\!]{\mleft({#2}\mright)}%
\newcommand{\pbrcx}[1]{\left[ {#1} \right]}%
\newcommand{\Prob}[1]{\mathop{\mathbf{Pr}}\!\pbrcx{#1}}

\newcommand{\ceil}[1]{\left\lceil {#1} \right\rceil}
\newcommand{\floor}[1]{\left\lfloor {#1} \right\rfloor}

\newcommand{\cardin}[1]{\left| {#1} \right|}%

\renewcommand{\th}{th\xspace}

\renewcommand{\Re}{\mathbb{R}}%
\usepackage[inline]{enumitem}

\newlist{compactenumA}{enumerate}{5}%
\setlist[compactenumA]{topsep=0pt,itemsep=-1ex,partopsep=1ex,parsep=1ex,%
   label=(\Alph*)}%

\newlist{compactenuma}{enumerate}{5}%
\setlist[compactenuma]{topsep=0pt,itemsep=-1ex,partopsep=1ex,parsep=1ex,%
   label=(\alph*)}%

\newlist{compactenumI}{enumerate}{5}%
\setlist[compactenumI]{topsep=0pt,itemsep=-1ex,partopsep=1ex,parsep=1ex,%
   label=(\Roman*)}%

\newlist{compactenumi}{enumerate}{5}%
\setlist[compactenumi]{topsep=0pt,itemsep=-1ex,partopsep=1ex,parsep=1ex,%
   label=(\roman*)}%

\newlist{compactitem}{itemize}{5}%
\setlist[compactitem]{topsep=0pt,itemsep=-1ex,partopsep=1ex,parsep=1ex,%
   label=\bullet}%

\providecommand{\Mh}[1]{#1}%

\newcommand{\PS}{\Mh{{P}}}%

\newcommand{\IIS}{\Mh{\mathsf{J}}}%
\newcommand{\IJS}{\Mh{\mathsf{K}}}%

\newcommand{\IS}{\Mh{\EuScript{I}}}%
\newcommand{\XS}{\Mh{\EuScript{X}}}%
\newcommand{\YS}{\Mh{\EuScript{Y}}}%

\newcommand{\mX}[1]{\Mh{\mu}(#1)}
\newcommand{\mmX}[1]{\Mh{\mu}\pth{#1}}
\newcommand{\mY}[2]{\Mh{\mu}_{#2}(#1)}

\newcommand{\crX}[1]{\Mh{\rho} \pth{#1}}%
\newcommand{\crY}[2]{\Mh{\rho}_{#2} \pth{#1}}%

\newcommand{\GG}{\Mh{\mathcal{G}}}%
\newcommand{\eps}{\varepsilon}%

\newcommand{\opt}{\Mh{\EuScript{O}}}
\newcommand{\optB}{\Mh{\EuScript{O}'}}

\newcommand{\optX}[1]{\Mh{\EuScript{O}}_{#1}}
\newcommand{\optY}[2]{\Mh{\EuScript{O}}_{#1}\pth{#2}}

\newcommand{\OB}{\Mh{\mathcal{P}}}%
\newcommand{\XX}{\Mh{\mathcal{X}}}%

\newcommand{\IntervalX}[1]{\Mh{\mathcalb{#1}}}
\renewcommand{\IJ}{\IntervalX{j}}%
\newcommand{\IK}  {\IntervalX{k}}%
\newcommand{\IL}  {\IntervalX{l}}%

\newcommand{\II}  {\IntervalX{t}}%
\newcommand{\IM}  {\IntervalX{m}}%
\newcommand{\IO}  {\IntervalX{o}}%
\newcommand{\IG}  {\IntervalX{g}}%

\newcommand{\IIL}  {\II_{\Mh{\leftarrow}}}%
\newcommand{\IIR}  {\II_{\Mh{\rightarrow}}}%

\newcommand{\uu}{\tau}%

\newcommand{\mvY}[2]{\Mh{\triangle}\pth{#1,#2}}%
\newcommand{\miY}[2]{\Mh{\nabla}^{}_{\!#2}#1}%
\newcommand{\miX}[1]{\Mh{\nabla}#1}%
\newcommand{\Li}{\Mh{\Xi}}%
\newcommand{\run}{\Mh{\Gamma}}%
\newcommand{\balanceX}[1]{\Mh{\chi}\pth{#1}}
\newcommand{\LX}[1]{\Mh{\mathsf{L}}_{#1}}
\newcommand{\RX}[1]{\Mh{\mathsf{R}}_{#1}}

\newcommand{\vopt}{\Mh{\nu}}%
\newcommand{\voptX}[1]{\vopt_{#1}}
\newcommand{\voptY}[2]{\vopt_{#1}\pth{#2}}

\DefineNamedColor{named}{AlgorithmColor}{cmyk}{0.07,0.90,0,0.34}

\newcommand{\Family}{\Mh{\mathcal{F}}}%
\newcommand{\AlgorithmI}[1]{{%
      \textcolor[named]{AlgorithmColor}{\texttt{\bf{#1}}}%
   }}

\DefineNamedColor{named}{OliveGreen} {cmyk}{0.64,0,0.95,0.40}
\providecommand{\ComplexityClass}[1]{{{\textcolor[named]{OliveGreen}{%
            \textsc{#1}}}}}
\providecommand{\NPHard}{{\ComplexityClass{NP-Hard}}%
   \index{NP!hard}\xspace}

\providecommand{\NP}{\ComplexityClass{NP}%
   \index{NP}\xspace}

\providecommand{\POLYT}{\ComplexityClass{P}\xspace}

\definecolor{darkred}{rgb}{0.1,0,0}

\newcommand{\ProblemC}[1]{{\textsf{\textcolor{darkred}{#1}\index{problem!#1}}}}

\newcommand{\etal}{\textit{et~al.}\xspace}
\newcommand{\predX}[1]{\mathrm{pred}\pth{#1}}
\newcommand{\dpY}[2]{\mathcalb{d}\pth{#1, #2}}%
\newcommand{\dpInY}[2]{\mathcalb{d}_{\in}\pth{#1, #2}}%

\newcommand{\profitCA}{\Mh{\mathcalb{q}}}%

\newcommand{\profitC}{\Mh{\mathcalb{p}}}%
\newcommand{\profitX}[1]{\profitC_{#1}}
\newcommand{\profitY}[2]{\profitC_{#1}\pth{#2}}

\newcommand{\nI}{\Mh{m}}%
\newcommand{\nP}{\Mh{n}}%

\begin{document}

\title{On Competitive Permutations for Set Cover by Intervals}
\author{%
   Sariel Har-Peled%
   \and%
   Jiaqi Cheng%
}

\date{\today}
\maketitle

\begin{abstract}
    We revisit the problem of computing an optimal partial cover of
    points by intervals. We show that the greedy algorithm computes a
    permutation $\Pi = \pi_1, \pi_2,\ldots$ of the intervals that is
    $3/4$-competitive for any prefix of $k$ intervals. That is, for
    any $k$, the intervals $\pi_1 \cup \cdots \cup \pi_k$ covers at
    least $3/4$-fraction of the points covered by the optimal solution
    using $k$ intervals.

    We also provide an approximation algorithm that in $O(n + m/\eps)$
    time, computes a cover by $(1+\eps)k$ intervals that is as good as
    the optimal solution using $k$ intervals, where $n$ is the number
    of input points, and $m$ is the number of intervals (we assume
    here the input is presorted).

    Finally, we show a counter example illustrating that the optimal
    solutions for set cover do not have the diminishing return
    property -- that is, the marginal benefit from using more sets is
    not monotonically decreasing.  Fortunately, the diminishing
    returns does hold for intervals.
\end{abstract}

\section{Introduction}

In the \ProblemC{max $k$ cover problem}, the input is a ground set
$\PS$ of $n$ elements, and a family $\Family$ of $m$ subsets of $\PS$,
and an integer $k$. The task is to pick $k$ sets of $\Family$, that
maximize the total number of elements of $\PS$ covered.  This problem
is \NPHard \cite{gj-cigtn-79}, and the standard greedy algorithm, of
repeatedly picking the set covering the largest number of elements not
covered yet, has approximation ratio of $1-1/e$ \cite{hp-agapm-98}.
If one wants to cover all the elements of $\PS$ (but minimize the
number of sets used), this is the \ProblemC{set cover} problem. In
this case, the greedy algorithm provides a $(1+\ln{n})$ approximation.
It is known that (essentially) no better approximation is possible
\cite{ds-aapr-13,ds-aapr-14} unless $\POLYT = \NP$.

\paragraph{Diminishing returns.}
A natural property of the greedy solution is that the benefit of the
$i$\th set in the cover declines as $i$ increases -- that is, the
$i$\th greedy set covers no more elements than the previous sets. This
phenomena is known as \emph{diminishing returns}. Somewhat
surprisingly, this phenomena does not hold for the optimal partial
solutions. Specifically, the $k$\th marginal value, is the increase in
coverage as one moves from the optimal $(k-1)$-cover to the optimal
$k$-cover. For the general set cover problem, the sequence of marginal
values is not monotone, as we show in \secref{d:r:optimal}.

\paragraph{Interval scheduling.}
A variant of the problem is where the ground set is on the real line,
and the sets are intervals. This problem rises naturally in scheduling
(i.e., time is the $x$-axis), known as \ProblemC{interval
   scheduling}. See \cite{lt-ois-94,sc-oaisp-1999,cjst-oselj-07,
   DBLP:journals/corr/abs-0812-2115} for related work.

\paragraph{Interval cover.}
The set cover problem becomes significantly easier if one is
restricted to points and intervals on the real line. It is well known
that the greedy algorithm adding intervals to the covers from let to
right, always adding the one extending furthest to the right, computes
the optimal solution. The $k$-cover variant can be solved using
dynamic programming in $O(nk)$ time \cite{gkkss-sd-09, egk-pisc-13,
   lldl-omkic-20}.

Edwards \etal \cite{egk-pisc-13} studied the \ProblemC{partial
   interval cover} problem, where needs to cover a specified number of
points using minimum number of intervals.  They also show that that
the diminishing return property holds in this case for the optimal
solution(s). Using this together with partitioning, they provide an
$(1+\eps)$-approximation algorithm to the number of intervals. Their
algorithm runs in $O\bigl(\eps^{-1} (n + m) \bigr)$ time.

\paragraph{Competitive ratio for interval cover.}

A natural question is how to order the input sets (i.e., intervals in
our case) so that they provide the best coverage for any prefix of
this ordering. Conceptually, we are interested in ordering the
intervals by the ``usefulness'' of the coverage they provide. The
greedy algorithm naturally provides such an ordering, known as the
\emphw{greedy permutation}. In particular, when considering the first
$k$ intervals of the greedy permutation, and comparing it to the
optimal $k$-cover, what is this competitive ratio? What is the
competitive ratio when we consider all $k$?

Since the optimal cover is unstable, and changes as one increases the
number of intervals used, it is not a priori clear what is the best
ordering if one wants to minimize the competitive ratio.

\subsection*{Our results.}

We prove that the diminishing returns holds for intervals -- we were
unaware of the work by Edwards \etal \cite{egk-pisc-13} who already
proved it. Our proof is somewhat different, and we include it in the
\apndref{d:intervals}.

In \secref{comp:greedy}, we prove that the greedy permutation (for the
case of intervals) provides a competitive ratio of $3/4$. We provide
an example showing that this analysis tight. This compares favorably
with the general case, where the competitive ratio is
$1-1/e \approx 0.6321 \ll 3/4$.

We provide an $(1-\epsilon)$-approximate algorithm with a complexity
of $O(\frac{n}{\eps} + k\log (\eps k))$. The algorithm is also based
on the diminishing returns property, and the partitioning method has
been altered to adapt for a non-discrete model. In addition to the
approximate algorithm, our research provided a tight approximate ratio
of $0.75$ for the greedy $k$ interval cover, as a comparison to the
approximate algorithm.

\section{Preliminaries}

\subsection{Definitions and problem statement}

\paragraph{Notations.}

In the following we deal with set systems. A set system $(\PS,\IS)$
has a ground set $\PS$ (which is finite in our case), and a family
$\IS$ of subsets of $\PS$. For a family of sets $\XS \subseteq \IS$,
let $\cup \XS = \cup_{X \in \XS} X$.  For a set $Y \subseteq \PS$, let
\begin{equation*}
    \XS \sqcap Y
    =%
    \Set{ x \in Y}{ x \in \cup \XS } \subseteq \PS.
\end{equation*}
In the following, given a set $X$, and an element $x$, we use the
notation $X - x = X \setminus \{x\}$, and similarly
$X+x = X \cup \{x\}$.

For a set $X$, its \emphi{measure} is $\mX{X} = |X \cap \PS|$.  For a
set of intervals $\IS$, its \emph{measure} is
$\mX{\IS} = \mX{ \cup \IS }$.

\begin{defn}
    An instance of \emphi{interval cover} is a pair
    $\IIS = (\PS, \IS)$, where $\PS \subseteq \Re$ is a set of points on
    the real line, and $\IIS$ is a set of intervals.
\end{defn}

\paragraph{The input.}
The input is an instance $\IIS = (\PS, \IS)$ of interval cover, where
$\nP = \cardin{\PS}$ and $\nI = \cardin{\IS}$.  Specifically, let
$\IS=\{\II_1, \II_2, \ldots, \II_\nI\}$ be the set of intervals,
sorted from left to right by their right endpoints (for simplicity of
exposition, we assume all endpoints are distinct).  Furthermore, we
assume that no two intervals contains the same subset of points of
$\PS$, and that no interval is contained inside another interval (as
one would also use the bigger interval in a cover, and the smaller
interval is as such redundant).  Thus, the order of the intervals by
their right endpoints, or by their left endpoints, is the same.

For two intervals $\II$ and $\IJ$, let $\II \prec \IJ$ indicates that
the left endpoint of $\II$ is to the left of the left endpoint of
$\IJ$ -- that is, $\II$ is to the \emphw{left} of $\IJ$.

\begin{defn}
    \deflab{opt:solution}%
    Let $\IIS = (\PS, \IS)$ the given instance of interval cover.  A
    set $\YS \subseteq \IS$ is \emphi{optimal $k$-cover}, if
    $|\YS| = k$, and the measure of $\YS$ is maximum among all such
    sets of intervals of size $k$. Let
    $\optX{k} = \optY{k}{\IIS} =\optY{k}{\PS} $ denote such an optimal
    $k$-cover of the point set $\PS$.  Let
    $\voptY{k}{\IIS} = \mX{\optY{k}{\IIS}}$.
\end{defn}

\paragraph{Problem definition.}

For a permutation $\pi$ of the intervals of $\IS$, its
\emphw{$k$-competitive ratio}, for any $k >0$, is
\begin{equation*}
    \crY{\pi}{k}
    =%
    \frac{\mX{ \cup_{i=1}^k \II_{\pi(i)} } }
    {\voptX{k}},
    \qquad\text{where}\qquad \voptX{k} = \mX{\optX{k}}.
\end{equation*}
Its overall \emphi{competitive ratio} is
$\crX{\pi} = \min_k \crY{\pi}{k}$.  The task at hand is to compute the
permutation with competitive ratio as close to one as possible.

\paragraph{The greedy algorithm.}
This algorithm repeatedly picks the interval that covers the most
points not yet covered, and add it to the current cover. Let
$\IG_1, \IG_2, \ldots$ be the input intervals as ordered by the greedy
algorithm. If there are several candidate intervals that cover the
same number of points, the greedy algorithm always pick the leftmost
such interval.

\paragraph{Diminishing returns and submodularity.}

An important property of the greedy algorithm is that the contribution
of each added set decreases.

\begin{defn}
    \deflab{marginal:value}%
    For a set of intervals $\XS$, and an interval $\II \in \XS$, its
    \emphi{marginal value} is
    \begin{equation*}
        \mvY{\II}{\XS}%
        =%
        \mX{ \XS } - \mX{\XS - \II}.
    \end{equation*}
\end{defn}

\begin{defn}
    \deflab{marginal:profit}%
    The \emphi{$i$\th marginal profit} of an instance
    $\IIS = (\PS, \IS)$ is the added value to the optimal solution by
    increasing the optimal solution to be of size $i$. Formally, it is
    the quantity
    \begin{equation*}
        \profitY{i}{\IIS}
        =
        \voptY{i}{\IIS} - 
        \voptY{i-1}{\IIS} 
    \end{equation*}
    see \defref{opt:solution}. Observe that
    $\voptY{k}{\IIS} = \sum_{i=1}^k \profitY{i}{\IIS}$.
\end{defn}

The following straightforward lemma shows that diminishing returns
property holds for the greedy solution.
\begin{lemma}
    Consider any instance of set cover (not necessarily of points and
    intervals).  Let $\GG_i = \{ \IG_1, \ldots, \IG_i\}$ be the prefix
    of the first $i$ intervals computed by the greedy algorithm. For
    all $i$, we have the \emphi{diminishing returns} property that
    \begin{equation*}
        \mX{\GG_i} - \mX{\GG_{i-1}} \geq 
        \mX{\GG_{i+1}} - \mX{\GG_{i}}.
    \end{equation*}
\end{lemma}
\begin{proof}
    If the diminishing property fails, then
    $\mvY{\IG_i}{\GG_i} < \mvY{\IG_{i+1}}{\GG_{i+1}}$. This implies
    that $\GG_i - \IG_i + \IG_{i+1}$ would cover more elements than
    $\GG_i$, which is impossible as the greedy algorithm chooses
    $\IG_i$ as the set that covers the largest number of elements that
    are yet uncovered by $\GG_{i-1}= \GG_i - \IG_i$.
\end{proof}

Surprisingly, diminishing returns does not hold for the optimal
solution -- see \secref{d:r:optimal}. Our target function is
submodular in the sense that an interval $\II \in \IS$ has lesser
value as we add it into a bigger solution. Formally, $\mX{\cdot}$ is
\emphi{submodular} if
\begin{equation*}
    \forall \XS, \YS, \II%
    \qquad%
    \XS \subseteq \YS \subseteq \IS, \text{ and } \II \in \IS \qquad
    \mX{ \XS + \II } - \mX{ \XS } %
    \geq%
    \mX{ \YS + \II } - \mX{ \YS }.    
\end{equation*}

\section{Competitiveness of the greedy algorithm}
\seclab{comp:greedy}

\subsection{Extremality and allowable patterns}

\begin{observation}
    \obslab{extremal}%
    (A) For an optimal solution $\optX{t}$, and for any two distinct
    intervals $\II, \IJ \in \optX{t}$ that intersects, one can assume
    that they are extremal to each other. Specifically, if
    $\II \prec \IJ$, then one can assume that $\IJ$ is the right most
    interval that intersect $\II$. This can be enforced by applying a
    greedy replacement of intervals on the optimal solution from left
    to right. Similarly, one can assume that $\II$ is the leftmost
    interval that intersects $\IJ$. An optimal solution that has this
    property is \emphi{extremal}.
\end{observation}

From this point on, we assume that all optimal solutions under
discussion are extremal.

\begin{lemma}
    \lemlab{only:two}%
    Let $\optX{t}$ be an (extremal) optimal solution, and let $\II$ be
    an interval not in $\optX{t}$. Then, $\II$ intersects at most two
    intervals of $\optX{t}$.
\end{lemma}

\begin{proof}
    Any interval (that is not $\II$) that intersects $\II$ must cover
    one of its endpoints, as no intervals contain each other. As such,
    if $\II$ intersects three intervals of $\optX{t}$, then two of
    them, say $\IJ$ and $\IK$, must cover one of the endpoints of
    $\II$ (say the left one).  Furthermore, assume that
    $\IJ \prec \IK$, see \figref{three}.
    \begin{figure}[h]
        \centerline{\includegraphics{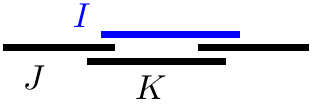}}
        \caption{}
        \figlab{three}
    \end{figure}%

    \noindent%
    But then, as $\II \notin \optX{t}$, it must be that $\IK \neq \II$
    and $\IK \prec \II$. This contradicts the extremality property for
    $\IJ\in \optX{t}$, as one can replace $\IK$ by $\II$ in the
    optimal solution.
\end{proof}

\begin{observation}
    \obslab{two}%
    The extremality property implies that for any two intersecting
    intervals $\II, \IJ$ in an optimal solution $\optX{t}$, with
    $\II \prec \IJ$, we have that no other interval $\IK \in \optX{t}$
    intersects $\IJ$, and $\IK \prec \IJ$.
\end{observation}

\begin{lemma}
    \lemlab{allowed:pattern}%
    Consider two optimal extremal solutions $\opt$ and $\OB$, and let
    $\II, \IJ, \IK \in \opt \cup \OB$ be three consecutive intervals
    such that $\II \prec \IJ \prec \IK$ and
    $\II \cap \IK \neq \emptyset$. Then, there are only two
    possibilities:
    \begin{compactenumI}
        \smallskip%
        \item $\II \in \opt \setminus \OB$,
        $\IJ \in \OB\setminus \opt$, and $\IK \in \opt \setminus \OB$,
        or

        \smallskip%
        \item $\II \in \OB \setminus \opt$,
        $\IJ \in \opt \setminus \OB$, and
        $\IK \in \OB \setminus \opt$.
    \end{compactenumI}
\end{lemma}
\begin{proof}
    The proof is by straightforward case analysis:
    \begin{compactenumi}
        \smallskip%
        \item $\II, \IJ \in \opt$.  This is impossible as $\IK$ can
        replace $\IJ$ in $\opt$, contradicting the extremality of
        $\opt$.
    \end{compactenumi}
    \smallskip%
    This implies that the following cases are impossible by symmetry:
    \begin{compactenumi}[resume]
        \smallskip%
        \item $\II, \IJ \in \OB$.

        \smallskip%
        \item $\IJ, \IK \in \opt$.

        \smallskip%
        \item $\IJ, \IK \in \OB$.
    \end{compactenumi}
    \smallskip%
    This readily implies that it is impossible that
    $\II \in \opt \cap \OB$, and the same holds for $\IJ$ and $\IK$.
    Thus, the only remaining possibilities are the ones stated in the
    lemma.
\end{proof}

\subsection{Competitive ratios}

The following is well known and is included for the sake of
completeness.
\begin{lemma}
    The competitive ratio of the greedy algorithm for set cover is
    $\geq 1-1/e$.
\end{lemma}
\begin{proof}
    Let $\optX{k}$ denote an optimal solution of size $k$.  In the
    beginning of the $i$\th iteration, let
    $\Delta_i = \voptX{k} - \mX{\GG_{i-1}}$ be the deficit. There must
    be a set in $\optX{k}$ that covers at least $\Delta_i/k$ elements
    that are not covered by the first $i-1$ greedy sets. As such, the
    greedy algorithm picks a set that cover at least this number of
    elements (and potentially many more). As such, we have
    \begin{align*}
      \Delta_{i+1}%
      &=%
        \voptX{k} - \mX{\GG_i}
        =%
        \voptX{k} - \mX{\GG_{i-1}}
        - \mX{ \IG_i \setminus \cup\GG_{i-1}}
        \leq%
        \Delta_i - \Delta_i/k
        =%
        (1-1/k)\Delta_i
      \\&%
      \leq
      (1-1/k)^{i}\Delta_1
      =%
      (1-1/k)^{i}\voptX{k}.
    \end{align*}
    As such, we have
    \begin{equation*}
        \crX{k}%
        =%
        \frac{\mX{ \cup_{i=1}^k \IG_{k}  } }{\voptX{k}}
        =%
        \frac{\voptX{k} - \Delta_{k+1}}{\voptX{k}}
        \geq %
        \frac{\voptX{k} - (1-1/k)^k\voptX{k}}{\voptX{k}}
        \geq%
        1 - \frac{1}{e},
    \end{equation*}
    since $1-x \leq \exp(-x)$, for $x \geq 0$.
\end{proof}

\begin{lemma}
    \lemlab{first}%
    (A) Consider an optimal $k$ covering
    $\opt = \optX{k} = \{\IO_1,\IO_2,\ldots, \IO_k\} \subseteq
    \IS$. For any interval $\IO \in \opt$, we have that $\opt - \IO$
    is an optimal cover of $\PS \setminus \IO$ by $k-1$ intervals.

    (B) The set $\{ \IO_{q+1}, \ldots, \IO_{k} \}$ is an optimal cover
    by $k-q$ intervals of $\PS \setminus \bigcup_{i=1}^q \IO_i$.
\end{lemma}

\begin{proof}
    (A) Let $\optB$ be an optimal cover of $\PS \setminus \IO$ by
    $k-1$ intervals.  If
    \begin{equation*}
        |\optB \sqcap (\PS \setminus \IO)| > |(\opt-\IO) \sqcap (\PS - \IO)|,
    \end{equation*}
    then $\optB + \IO$ covers more points of $\PS$ than $\opt$, which
    is a contradiction to the optimality of $\opt$. %

    (B) Follows by repeated application of (A).
\end{proof}

We remind the reader that the set of input intervals $\IS$ is made out
intervals $\II_1, \ldots, \II_\nI$, and they are sorted in increasing
order by their left endpoint.

\begin{lemma}
    \lemlab{one:or:two}%
    Let $\II \in \IS$ be an interval that covers the maximum number of
    points of $\PS$ among all the intervals of $\IS$. The interval
    $\II$ is either (i) in one of the optimal $k$ covers, or (ii)
    alternatively, exactly two intervals $\IJ, \IJ' \in \IS$ of an
    optimal cover overlap it, where $\IJ \prec \II \prec \IJ'$.
\end{lemma}

\begin{proof}
    If $\II$ appears in an optimal $k$ cover, then we are done. So
    assume $\II$ is not in any optimal cover, for all optimal
    $k$-covers. Fix such an optimal cover $\opt = \optX{k}$. If the
    right endpoint of $\II$ is covered by two intervals of $\opt$ then
    one of them can be replaced by $\II$, and yields an equivalent
    solution, a contradiction. Similarly, if the left endpoint of
    $\II$ is covered by two intervals in $\opt$, the same argument
    applies. As $\II$ can not contain fully any interval of $\opt$ it
    follows that it intersects at most two such intervals.

    If there was only one such interval, then one could just replace
    this interval by $\II$, yielding an equivalent or better solution,
    which would imply (i).

    The property that $\IJ \prec \II \prec \IJ'$ readily follows, as
    one of the two intervals covers the left endpoint of $\II$, and
    the other one covers the right endpoint of $\II$.
\end{proof}

\begin{lemma}
    \lemlab{half}%
    We have $\mX{\GG_k} \geq \mX{\opt_{2k}}/2$.
\end{lemma}
\begin{proof}
    If $k=1$ the claim is immediate, as
    $\mX{\opt_{2}} \leq 2 \mX{\opt_{1}} = 2 \mX{\GG_1}$.  Consider the
    optimal $2k-2$ cover $\opt'_{2k-2}$ of $\PS - \IG_1$. There are
    several possibilities to consider:
     
    If $\IG_1 \in \opt_{2k}$: Let $\IO_j$ be to any interval in
    $\opt_{2k} \setminus \{\IG_1 \}$. Set $\opt''_{2k-2}$ to be the
    optimal $2k-2$ cover of $\PS - \IG_1 - \IO_j$.  By the optimality
    of $\opt'_{2k-2}$, we have
    \begin{equation*}
        \mY{\opt'_{2k-2}}{\PS-\IG_1}
        \geq
        \mY{\opt''_{2k-2}}{\PS-\IG_1-\IO_j}.
    \end{equation*}
    Recall that $2\mX{\IG_1}$ is larger or equal to the coverage
    provided by any two intervals of $\opt_{2k}$. By induction on the
    point set $\PS-\IG_1$, we have by \lemref{first} that
    \begin{align*}
      \mX{\optX{2k}}
      &=
        \mX{\IG_1 \cup \IO_j} + \mY{\opt''_{2k-2}}{\PS-\IG_1-\IO_j}
      \leq%
      2\mX{\IG_1} + \mY{\opt'_{2k-2}}{\PS-\IG_1}
      \\&%
      \leq%
      2\mX{\IG_1} + 2\mathrm{greedy}(\PS - \IG_1, k-1)
      =%
      2\mX{\GG_k}.
    \end{align*}

    Otherwise, $\IG_1 \notin \optX{2k}$. %
    By \lemref{only:two}, $\IG_1$ intersects at most two intervals of
    $\opt_{2k}$.  Let $\IO_1$ and $\IO'_1$ be these two intervals. Set
    $\opt'' = \opt_{2k} - \IO_1 - \IO_1'$. By \lemref{first}, $\opt''$
    is an optimal $2k-2$ cover of $P-\IO_1-\IO'_1$, and by
    construction is does not cover any point of $\IG_1$.  %
    As such
    \begin{equation*}
        \mY{\opt'_{2k-2}}{\PS-\IG_1}%
        \geq%
        \mY{\opt''}{\PS-\IG_1}
        =%
        \mY{\opt''}{\PS}
        \geq%
        \mY{\opt''}{\PS-\IO_1-\IO'_1}.
    \end{equation*}
    As $2\mX{\IG_1} \geq \mX{\IO_1 \cup \IO'_1}$, by induction we have
    \begin{align*}
      \mX{\optX{2k}}%
      &=%
        \mX{\IO_1 \cup \IO'_1} + \mY{\opt''}{\PS-\IO_1-\IO'_1}%
      \leq%
      2\mX{\IG_1} + \mY{\opt'_{2k-2}}{\PS-\IG_1}%
      \\&%
      \leq%
      2\mX{\IG_1} + 2\mathrm{greedy}(\PS - \IG_1, k-1)%
      =2\mX{\GG_k}.
    \end{align*}                             
\end{proof}

\subsubsection{The even case}

\begin{lemma}
    \lemlab{even}%
    If $k$ is even, we have $\mX{\GG_k} \geq 3\voptX{k}/4$, where
    $\voptX{k}$ is the optimal coverage by $k$ intervals.
\end{lemma}

\begin{proof}
    Break the greedy $k$ cover into two parts
    $\GG' = \{ \IG_1,\ldots, \IG_{k/2}\}$ and
    $\GG'' = \{\IG_{k/2+1}, \ldots, \IG_{k}\}$.  By \lemref{half},
    $\mX{\GG'}/\voptX{k} \geq 1/2$.

    Let $\PS' = \PS \setminus \bigcup\GG'$.  Observe that $\GG''$ is
    the greedy $k/2$ cover of $\PS'$. Now, the optimal $k$ cover of
    $\PS'$ has value at least $\voptX{k}-\mX{\GG'}$. By \lemref{half},
    we have $\mY{\GG''}{\PS'} \geq (\voptX{k}-\mX{\GG'})/2$. As such,
    we have
    \begin{align*}
      \frac{\mX{\GG_{k}}}{\voptX{k}}%
      &=%
        \frac{ \mX{\GG'} + \mY{\GG''}{\PS'} }{\voptX{k}}%
        \geq%
        \frac{\mX{\GG'} + (\voptX{k}-\mX{\GG'})/2}{\voptX{k}}%
        =%
        \frac{ \voptX{k}/2 + \mX{\GG'}/2 }{\voptX{k}}%
        \geq%
        \frac{1}{2} + \frac{ \mX{\GG'}}{2\voptX{k}}        
        \geq
        \frac{3}{4},
    \end{align*}
    since, by \lemref{half}, $\mX{\GG'}/\voptX{k} \geq 1/2$.
\end{proof}

\subsubsection{The odd case}
\newcommand{\greedyA}{\AlgorithmI{greedy}}%

\begin{lemma}
    \lemlab{onefour}%
    $\mX{\GG_{k+1}} \geq \voptX{2k+1}/2+\mX{\IG_1}/4$.
\end{lemma}

\begin{proof}
    Consider the interval $\IG_1$, by greedyness, it is the interval
    that covers the most points in the input $\IS$.

    First consider the case that $\IG_1 \in \optX{2k+1}$. After
    removing $\IG_1$ from $\PS$, the intervals
    $\IG_2, \ldots, \IG_{k+1}$ are the greedy $k$ cover of
    $\PS - \IG_1$.  For
    $\voptX{2k+1} = \voptY{2k+1}{\PS} = \mX{\optY{2k+1}{\PS}}$,
    observe that
    $\voptX{2k+1}- \mX{\IG_1} \leq \voptY{2k}{\PS - \IG_1}$.  By
    \lemref{half} we have:
    \begin{equation*}
        \mX{\GG_{k+1}}
        =%
        \mX{\greedyA(\PS - \IG_1, k)} + \mX{\IG_1} 
        \geq%
        \frac{\voptX{2k+1}-\mX{\IG_1}}2+\mX{\IG_1}
        \geq%
        \frac{\voptX{2k+1}+\mX{\IG_1}}{2}.        
    \end{equation*}

    The other possibility is that $\IG_1 \notin \optX{2k+1}$. Then by
    \lemref{one:or:two}, $\IG_1$ intersects at most two intervals
    $\IO_1$ and $\IO'_1$, where $\IO_1, \IO'_1 \in \optX{2k+1}$.  Let
    $\IM = (\IG_1 \cap \IO_1)\setminus \IO_1'$ and
    $\IM' = (\IG_1 \cap \IO_1')\setminus \IO_1$, and observe that they
    both contained in $\IG_1$, and are disjoint. Assume, with loss of
    generality, that $\mX{\IM'} \leq \mX{\IM}$. This implies that
    $\mX{\IM'} \leq \mX{\IG_1}/2$. Thus, we have
    \begin{align*}
      \mmX{\bigl.\cup \optX{2k+1} \setminus (\IG_1 \cup \IO_1)}
      &=%
        \mmX{\bigl.\cup \optX{2k+1} \setminus \pth{ (\IG_1 \cap \IO_1) \cup
        (\IG_1 \cap
        \IO_1')  \cup \IO_1 \bigr.}\Bigr. }
      \\&%
      =%
      \voptX{2k+1} - \mX{\IO_1} - \mX{\IM'}
      \geq%
      \voptX{2k+1} - \mX{\IG_1} - \frac{\mX{\IG_1}}{2}
      =
      \voptX{2k+1} - \frac{3}{2}\mX{\IG_1},
    \end{align*}
    as $\mX{\IG_1} \geq \mX{\IO_1}$.

    \remove{%
       \begin{align*}
         &\mmX{\bigl.\cup \optX{2k+1} - (\IG_1 \cup \IO_1)}
           =%
           \voptX{2k+1} - \mX{\IO_1}
           - \mmX{\bigl.(\IG_1 \cap \IO'_1)\setminus \IO_1}.
         \\
         \text{and}\quad
         &
           \mX{\cup \optX{2k+1} - (\IG_1 \cup \IO'_1)}
           =%
           \voptX{2k+1} - \mX{\IO'_1}
           - \mX{(\IG_1 \cap \IO_1)\setminus \IO'_1}.
       \end{align*}%
    }%

    Thus when the points covered by $\IG_1$ are removed from $\PS$,
    the union size of the remaining optimal $2k$ cover is at least
    $\mX{\cup \optX{2k+1} - (\IG_1 \cup \IO_1)}/2$. By \lemref{half},
    we have
    \begin{equation*}
        \mmX{\bigl.\greedyA(\PS-\IG_1, 2k)}%
        \geq
        \frac{\mX{\optY{2k}{\PS - \IG_1}\bigr.}}{2}
        \geq%
        \frac{\mX{\cup \optX{2k+1} - (\IG_1 \cup \IO_1)}}{2}%
        \geq%
        \frac{\voptX{2k+1}}{2} - \frac{3\mX{\IG_1}}{4}.
    \end{equation*}
    As such, we have
    \begin{equation*}
        \mX{\GG_{k+1}}%
        =%
        \mX{\IG_1}+\mmX{\greedyA(\PS - \IG_1, k)\bigr.} %
        \geq%
        \mX{\IG_1}+  \frac{\voptX{2k+1}}{2}
        - \frac{3\mX{\IG_1}}{4}\geq \frac{\voptX{2k+1}}{2} +
        \frac{\mX{\IG_1}}{4}.            
    \end{equation*}
\end{proof}

\begin{defn}
    For a set of intervals $\XS$, and an interval $\II \in \XS$, its
    \emphi{marginal interval} is
    \begin{equation*}
        \miX{\II}%
        =%
        \miY{\II}{\XS}%
        =%
        \II - \bigcup (\XS \setminus \{\II\}).
    \end{equation*}
\end{defn}

\begin{lemma}
    \lemlab{2marginal}%
    For a set of intervals $\XS\subset\IS$, such that no interval of
    $\IS$ contains another interval of $\IS$, any interval
    $\II \in \IS$ intersects at most two marginal intervals of $\XS$.
\end{lemma}

\begin{proof}
    The marginal intervals of $\XS$ are disjoint. As such, an interval
    $\II$ intersecting three marginal intervals
    $\IM_1 \prec \IM_2 \prec \IM_3$, would have to contain the
    original interval of $\XS$ inducing $\IM_2$, which is impossible.
\end{proof}

\begin{lemma}
    \lemlab{odd}%
    For any non-negative integer $k$, we have
    $\mX{\GG_{2k+1}} \geq 3\voptX{2k+1}/4$.
\end{lemma}

\begin{proof}
    If $k = 0$, the claim is immediate because
    $\mX{\GG_{1}} = \voptX{1}$.

    If $k \geq 1$, we separate the greedy $2k+1$ cover into two parts:
    $\GG_k = \{\IG_1,\ldots, \IG_k\}$ and
    $\GG'' = \{ \IG_{k+1},\ldots,\IG_{2k+1}\}$. By \lemref{2marginal},
    we have that $\IG_1, \ldots, \IG_k$ intersect at most $2k$
    marginal intervals of
    $\optX{2k+1}$. Let $\II$ be the interval of $\optX{2k+1}$, such
    that $\IG_1,\ldots, \IG_k$ do no overlap its marginal interval
    $\IM = \miY{\II_x}{\optX{2k+1}}$.
    
    By \lemref{half}, we have that
    \begin{equation}
        \eqlab{12a}
        \mX{\GG_k}%
        \geq%
        \frac{\voptX{2k}}{2}%
        \geq%
        \frac{\voptX{2k+1} - \mvY{\II}{\optX{2k+1}}}{2}
        =%
        \frac{\voptX{2k+1} - \mX{\IM}}{2}.
    \end{equation}

    Let $\PS' = \PS - \bigcup\GG_k$.  Since $\IM$ does not intersect
    any intervals in $\GG_k$, we have that the largest of the
    remaining intervals over $ \PS'$ is at least of size
    $\mvY{\II_x}{\optX{2k+1}}$. Thus by \lemref{onefour},
    \begin{align*}
      \beta%
      =%
      \mmX{\greedyA(\PS', k+1)\bigr.} %
      \geq%
      \frac{\voptY{2k+1}{\PS'}}{2} + \frac{\mX{\IG_{k+1}}}{4}
      \geq%
      \frac{\mY{\opt_{2k+1}}{\PS'}}{2} +
      \frac{\mX{\IM}}{4}
      \geq
      \frac{\voptX{2k+1} -
      \mX{\GG_k}}2 +
      \frac{\mX{\IM}}{4}.
    \end{align*}
    Thus, by \Eqref{12a}, we have
    \begin{equation*}
        \mX{\GG_{2k+1}}
        =%
        \mX{\GG_k} + \beta
        \geq%
        \frac{\voptX{2k+1}}2 + \frac{\mX{\GG_k}}{2}
        +\frac{\mX{\IM}}{4}
        \geq%
        \frac{\voptX{2k+1}}{2} + \frac{\voptX{2k+1} -
           \mX{\IM}}{4}
        +
        \frac{\mX{\IM}}{4}
        \geq
        \frac{3}{4} \voptX{2k+1}.
    \end{equation*}
\end{proof}

\subsection{The result}

\begin{theorem}
    The greedy algorithm for cover by intervals has competitive ratio
    at least $3/4$, for any prefix of the greedy permutation computed
    by the algorithm.
\end{theorem}
\begin{proof}
    Combining \lemref{odd} and \lemref{even}, implies that for any
    $k$, we have $\mX{\GG_k} \geq (3/4)\voptX{k}$.
\end{proof}

\begin{lemma}
    The $3/4$ competitive ratio of the greedy algorithm is tight.
\end{lemma}

\begin{proof}
    Consider three intervals $\II_1$, $\II_2$, $\II_3$, where
    $\mX{\II_1} = \mX{\II_2} = s$, and $\mX{\II_3} = s + \eps$. In
    this example, $\II_1$ and $\II_2$ are connected end-to-end. Let
    the connecting point of $\II_1$ and $\II_2$ be the median point of
    $\II_3$. As such, the greedy two cover would first include
    $\II_3$, then one of $\II_1$ and $\II_2$. The union size of greedy
    two cover is $\frac{3s + \eps}{2}$. The optimal two cover has a
    union size of $2s$. The competitive ratio is
    $\frac{3}{4} + \frac{\eps}{2s}$.
\end{proof}
\section{Diminishing returns do not hold for %
   optimal set cover}
\seclab{d:r:optimal}

\subsection{The construction}

Let $\uu > 1$ be some arbitrary integer.  In the following, pick some
arbitrary rational numbers $\alpha, \beta, \gamma \in (0,1)$, such
that $0 < \alpha < \beta < \gamma$,
\begin{equation*}
    \beta - \alpha < \gamma - \beta,
    \quad\text{ and }\quad%
    \frac{\gamma}{\uu+2} < \frac{\beta}{\uu+1} < \frac{\alpha}{\uu}.
\end{equation*}
We have a ground set $U$ -- this set is going to be a sufficiently
large finite set (more on that below).  For a set $X \subseteq U$, its
\emphi{measure} is
\begin{equation*}
    \mu( X) = \frac{|X|}{|U|}.
\end{equation*}

In the following we pick some sets from the ground set -- how exactly
we do that so that we have the desired properties listed below is
described in \secref{howto}.

We pick $\uu$ disjoint sets $B_1, \ldots, B_\uu$ from $U$, each one of
measure $\alpha/\uu$.

Next, we pick $\uu+1$ disjoint sets $C_1, \ldots, C_{\uu+1}$ each one
of measure $\beta/(\uu+1)$, such that for all $i,j$, we have
\begin{equation*}
    \mu( B_i \cap C_j) = \frac{ \alpha}{\uu(\uu+1)}.
\end{equation*}

Finally, we pick disjoint sets $D_1, \ldots, D_{\uu+2}$ each one of
measure $\gamma/(\uu+2)$, such that for all $i$ and $j$ we have
\begin{equation*}
    \mu( B_i \cap D_j) = \frac{ \alpha}{\uu(\uu+2)}.
\end{equation*}
Similarly, for all $i$ and $j$, we require that
\begin{equation*}
    \mu( C_i \cap D_j) = \frac{ \beta}{(\uu+1)(\uu+2)}.
\end{equation*}

\subsubsection{Realization in three dimensions}

The above construction can be realized in three dimensions using
axis-parallel cubes if one uses volume for measure.  So, assume
$\alpha =1$.  Consider the unit cube in three dimensions $[0,1]^3$. We
set
\begin{equation*}
    B_i = [(i-1)/\uu, i/\uu] \times [0,1]^2, \qquad \text{ for }
    i=1,\ldots, \uu.
\end{equation*}
That is, $B_1, \ldots, B_\uu$ slices the unit cube into equal boxes
along the $x$-axis.

Next, consider the enlarged cube
$[0,\beta] \times [0,1] \times [0,1]$. We set $y = 1/(\uu +1)$ and
\begin{equation*}
    C_i = [0,\beta] \times [(i-1)y,  i y] \times [0,1],
    \qquad \text{ for }
    i=1,\ldots, \uu+1.
\end{equation*}

This is illustrated in \figref{2d}.

\begin{figure}[h]
    \hfill%
    \includegraphics[page=1]{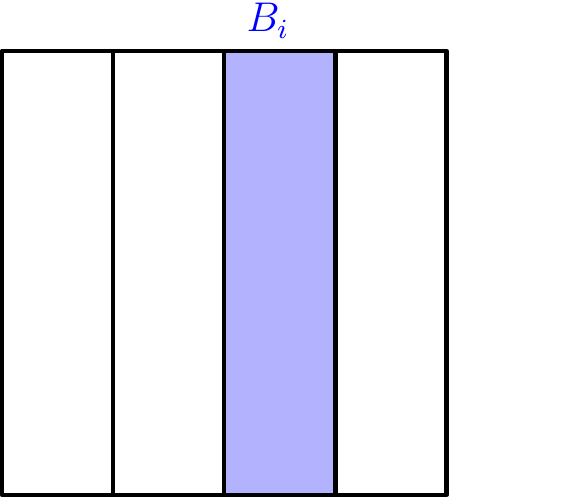}%
    \hfill%
    \includegraphics[page=2]{figs/2d_example}%
    \hfill\phantom{}

    \caption{}
    \figlab{2d}
\end{figure}

Finally, consider the cube
$[0,\beta] \times [0,\gamma/\beta] \times [0,1]$. We set
$z = 1 /(\uu+2)$ and
\begin{equation*}
    D_i = [0,\beta] \times [0,\gamma/\beta] \times
    [(i-1)z,
    i/z],
    \qquad \text{ for }
    i=1,\ldots, \uu+2.
\end{equation*}

It is easy to verify that this construction has the required
properties from above.

\subsection{Some properties}

Here is a list of some easy properties that the construction has:
\begin{compactenumI}
    \item For any $i$, we have that $\cup_j ( B_i \cap C_j) = B_i$.
    Indeed, since the $C_j$s are disjoint, we have
    \begin{equation*}
        \frac{\alpha}{k}%
        =%
        \mu( B_i ) %
        \geq %
        \sum_j \mu( B_i \cap C_j)%
        = %
        (\uu+1) \frac{\alpha}{\uu(\uu+1)}%
        =%
        \frac{\alpha}{k}.
    \end{equation*}
    
    \item $\cup_i B_i \subseteq \cup_j C_j$.

    \item Similarly, $\cup_j C_j \subseteq \cup_k D_k$.

\end{compactenumI}

\begin{lemma}
    The optimal cover by $\uu$ sets is $B_1, \ldots, B_{uu}$.
\end{lemma}
\begin{proof}
    Indeed, the sets $C_1, \ldots, C_{\uu+1}, D_1, \ldots, D_{\uu+2}$
    are smaller than $B_i$, for all $i$. Furthermore, the sets
    $B_1, \ldots, B_\uu$ are disjoint, which implies that it is indeed
    the largest possible cover by $\uu$ sets.
\end{proof}

\begin{lemma}
    The optimal cover by $\uu+1$ sets is $C_1, \ldots, C_{uu +1}$.
\end{lemma}
\begin{proof}
    Indeed, the set $D_k$ is smaller than $C_j$, for all $j$ and
    $k$. Furthermore, the sets $C_1, \ldots, C_{\uu+1}$ are disjoint,
    which implies that any optimal cover by $\uu+1$ sets can involve
    only $B$s and $C$s.

    Consider a set $C_j$, and observe that
    \begin{equation*}
        \mu( C_j \setminus \cup_i B_i)
        =%
        \frac{\beta}{\uu+1} - \sum_{i=1}^\uu \mu( B_i \cap C_j)
        =%
        \frac{\beta}{\uu+1}  - \uu \frac{\alpha}{\uu(\uu+1)}
        =%
        \frac{\beta - \alpha }{ \uu+1}
        >%
        0,
    \end{equation*}
    since $\beta > \alpha$. This implies that each $C_j$ contains
    elements that are not in any of the $B_i$s.
    
    As such, any cover by $\uu+1$ sets (made out of $B_i$s and $C_j$s)
    that does not include all sets of $C_1, \ldots, C_{\uu+1}$, must
    fail to cover some element in $\cup_j C_j$. This implies that
    $\cup_j C_j$ is an optimal cover.
\end{proof}

\begin{lemma}
    The optimal cover by $\uu+2$ sets is $D_1, \ldots, D_{uu+2}$.
\end{lemma}
\begin{proof}
    Consider a set $D_k$, and observe that
    \begin{equation*}
        \mu( D_k \setminus \cup_j C_j)
        =%
        \frac{\gamma}{\uu+2} - \sum_{j=1}^{\uu+1} \mu( C_j \cap D_k)
        =%
        \frac{\gamma}{\uu+2}  - (\uu+1) \frac{\beta}{(\uu+1)(\uu+2)}
        =%
        \frac{\gamma - \beta }{ \uu+2}
        >%
        0,
    \end{equation*}
    since $\gamma > \beta$. This implies that each $D_k$ contains
    elements that are not in any of the $C_j$s (and thus also elements
    not covered by any of the $B_i$s). As such, any other cover by
    $\uu+2$ sets fails to cover some element of $\cup_k D_k$, which
    being contained in this union, thus implying the claim.
\end{proof}

\subsection{How to pick the sets exactly}
\seclab{howto}

One can explicitly describe how to pick the sets, but we instead are
going to use an existential argument that is easier to see. Pick $n$
to be a sufficiently large, such that $\alpha n, \beta n, \gamma n$
are all integer numbers divisible by $\uu, \uu+1$ and $\uu+2$.  Let
$U =\{ 1,\ldots, n\}$.  Let $b = \alpha n / \uu$, and we set
$B_i = \{ (i-1)b + 1, \ldots, i b -1 \}$, for $i=1,\ldots, \uu$.

Next, we random assign each element of $\{1,\ldots, \beta n \}$ to
$C_1, \ldots, C_{\uu+1}$ with probability $1/(\uu+1)$. Standard
application of Chernoff's inequality implies that
\begin{equation*}
    \Prob{ \left| \mu(  C_j) - \frac{\beta}{\uu+1}  \right|
       > \sqrt{ \frac{c \log n}{n} } }%
    <%
    \frac{1}{n^{O(1)}}.       
\end{equation*}

Namely, by picking $n$ to be sufficiently large enough, we can assume
the measure of the $C_j$s are arbitrarily close to the desired
measure. Furthermore, using the same argumentation we have
\begin{equation*}
    \Prob{ \left| \mu( B_i \cap C_j) - \frac{1}{\uu(\uu+1)}  \right|
       > \sqrt{ \frac{c \log n}{n} } }%
    <%
    \frac{1}{n^{O(1)}}.       
\end{equation*}
Finally, we chose the sets $D_1, \ldots, D_{\uu+2}$, by assigning each
element of $\{1,\ldots, \gamma n \}$ to one of these sets with equal
probability. Again, a Chernoff type argument implies that all the
desired measures hold as desired within additive error that is
arbitrarily small. Picking all these additive errors to be smaller
than $(\beta-\alpha)/(\uu+2)^{10}$ (say), implies that the above
example implies the desired properties.

\section{Algorithms}

The input is a set of $n$ points $\PS \subseteq \Re$, and a set of $m$
intervals $\IS$. For simplicity of exposition, we assume that
$m = O(n)$. Here, the measure of the coverage is
$\mY{\II}{\PS} = \cardin{ \PS \cap \II}$.

\subsection{Dynamic programming algorithm}

\newcommand{\OptY}[2]{\mathrm{opt}\pth{#1,#2}}
\newcommand{\OptInY}[2]{\mathrm{opt}_{\in}\pth{#1, #2}}

The task at hand is to compute (exactly) the cover by $k$ intervals of
$\IS$ that maximizes the coverage

The algorithm starts by sorting, in $O(n\log n)$ time, all the
intervals in $\IS$ in non-decreasing order of their left endpoint.
For an interval $\II$.

\begin{defn}
    For an interval $\II$, let $\IIL$ be the first interval
    intersecting $\II$ in this sorted order.  Let $\predX{\II}$ be the
    predecessor of $\II$ in this order -- it is the interval
    immediately to the left of $\II$.
\end{defn}

Let $\IS_{\leq \II}$ be all the intervals of $\IS$ that are before
$\II$ in this order (including $\II$).  Let $\OptY{\II}{t}$ to be the
optimal cover by $t$ intervals of $\IS_{\leq \II}$.  Similarly, let
$\OptInY{\II}{t}$ to be the optimal cover by $t$ intervals of
$\IS_{\leq \II}$ that must contain $\II$ in the cover.

\begin{observation}    
    Consider an interval $\II \in \IS$, and the cover
    $\opt = \OptInY{\II}{t}$. Then, if there is an interval
    $\II^+ \in \opt$ to the left of $\II$ that intersects $\II$, then
    $\opt - \II^+ + \IIL$ is an equivalent solution providing the same
    coverage, as $\II^+ \subseteq \IIL \cup \II$.
\end{observation}

Note that $\IIL$ can be precomputed, in linear time, for each interval
$\II \in \IS$ by linear scanning (after the sorting).

When our algorithm is scanning from first interval to the last, it
would monotonically discard all the intervals to the left of current
$\IIL$.

After sorting and pre-processing comes to the dynamic programming
algorithm. Let $\dpInY{\II}{t} = \mX{\OptInY{\II}{t}}$ and
$\dpY{\II}{t} = \mX{\OptY{\II}{t}}$. We get the following recursive
definitions of these two quantities
\begin{align*}
  \dpInY{\II}{t}%
  &=%
    \max
    \begin{cases}
        \mX{\II} & t= 1\\%
        \dpInY{\IIL}{t-1}+\mX{\II \setminus \IIL} \qquad &t > 1 \text{
           and } \IIL \text{ is defined}
        \\
        \dpY{\predX{\IIL}}{t-1\bigr.} + \mX{\II} &t > 1 \text{ and
        }\predX{\IIL}\text{ is defined,}
    \end{cases}
  \\  
  \text{and } \qquad 
  \dpY{\II}{t}%
  &=%
    \max
    \begin{cases}
        \dpInY{\II}{t} \\
        \dpY{\predX{\II}}{t}.
    \end{cases}
\end{align*}

\begin{theorem}
    \thmlab{d:p:k:cover}%
    Given a set $\PS$ of $\nP$ points on the line, a set $\IS$ of
    $\nI$ intervals, and a parameter $k$, one can compute the optimal
    $k$-cover of $\PS$ by $k$ intervals of $\IS$, in $O(\nI k)$ time,
    assuming the intervals and points are presorted. Otherwise, the
    running time is $O(\nI k + (\nI+\nP) \log (\nI+\nP))$.

    Furthermore, the execution of this algorithm can be resumed, to
    compute the optimal $k+1$-cover in additional $O(\nI)$ time, and
    this can be done repeatedly, to compute an optimal $t$-cover, for
    any $t > k$, in $O((t-k)\nI)$ additional time.
\end{theorem}
\begin{proof}
    Observe that, by sweeping from left to right, one can compute for
    all intervals $\II \in \IS$, the quantities
    $\mX{\II} = \cardin{ \PS \cap \II}$ and
    $\mX{\II \setminus \IIL} = \cardin{\PS \cap \pth{\II \setminus
          \IIL}\bigr.}$. This takes $O(\nP+\nI)$ time. The algorithm
    then follows by using dynamic programming using the recursive
    formals provided above together with memoization.
\end{proof}

\subsection{An approximate algorithm for maximum $k$-coverage}

\subsubsection{Merging optimal solutions from disjoint instances.}

Given two instances of interval cover $\IIS=(\PS,\IS)$ and
$\IIS' = (\PS', \IS')$ that are disjoint (that is
$(\cup \IS) \cap (\cup \IS') = \emptyset$), consider computing the
optimal marginal profits  for the two instances. That is, for all $i$, let
\begin{equation*}
    \profitCA_{i} = \profitY{i}{\IIS}
    \qquad\text{ and }\qquad%
    \profitCA_{i}' = \profitY{i}{\IIS'}.
\end{equation*}
see \defref{marginal:profit}. By \thmref{diminishing:returns} the
sequences $\Pi \equiv \profitCA_{1}, \profitCA_{2},\ldots$ and
$\Pi' \equiv \profitCA_{1}', \profitCA_{2}',\ldots$ are
non-increasing.  As a reminder, the value of the optimal $k$-cover of
$\IIS$, is the sum of the first $k$ elements in $\Pi$.

Next, consider the merged instance
$\IJS = (\PS \cup \PS', \IS \cup \IS')$, and the sorted non-decreasing
sequence $\alpha_1 \geq \alpha_2 \geq \cdots$ formed by merging $\Pi$
and $\Pi'$.

\begin{claim}
    \clmlab{merge}%
    For any $i$, we have that $\alpha_i = \profitX{i}$, where
    $\profitX{i} = \profitY{i}{\IJS}$.
\end{claim}

\begin{proof}
    Consider the optimal $i$-cover of $\IJS$. It uses $\alpha$
    intervals of $\IIS$ and $\beta$ intervals of $\IIS'$, where
    $\alpha + \beta = i$. As such, the value of this optimal solution
    is
    \begin{equation*}
        \sum_{i=1}^\alpha \profitCA_i + 
        \sum_{i=1}^\beta \profitCA_i'.
    \end{equation*}
    If $\profitCA_{\alpha+1} > \profitCA_\beta$, then a better
    solution is formed by taking $\alpha+1$ intervals of $\IIS$ and
    $\beta-1$ intervals of $\IIS'$. Similarly, if
    $\profitCA_{\alpha} < \profitCA_{\beta+1}$, then a better solution
    is formed by taking $\alpha-1$ intervals of $\IIS$ and $\beta+1$
    intervals of $\IIS'$. It follows that the optimal strategy is to
    always partition the optimal solution into the two subproblems
    according to the merged (sorted) sequence, which implies the
    claim.
\end{proof}

\subsubsection{Settings}

The input is an instance $\IIS = (\PS, \IS)$ of interval cover, and
parameter $\eps$ and $k$. Here $\nP = |\PS|$ and $\nI = \cardin{\IS}$,
with $\IS = \{\II_1, \ldots, \II_\nI\}$.

The task at is to approximate the maximum coverage provided by $k$
intervals of $\IS$ -- specifically, the approximation algorithm would
output a cover with $(1+\eps)k$ intervals, that has value better than
optimal $k$-cover.  In the following, assume the intervals of $\IS$
are sorted from left to right by their left endpoint, and the points
of $\PS$ are sorted similarly.

\subsubsection{The algorithm}
Let $\Delta = \ceil{ m/ (\eps k)}$. For $K = \floor{\eps k }$, the
algorithm takes the set of intervals
\begin{equation*}
  \XS = \Set{\II_{j\Delta}}{j=1, \ldots,K}
\end{equation*}
into the computed cover. Since no input interval contains another
interval, this breaks the given instance into $K+1$ instances, where
the set of intervals for the $j$\th instance is
\begin{equation*}
    \IS_j = \{ \II_{j\Delta+1}, \ldots, \II_{(j+1)\Delta-1} \},
\end{equation*}
for $j=0,\ldots, K$.  The algorithm computes the set
$\PS' = \PS \setminus \cup\XS$. For each set of intervals $\IS_j$, the
corresponding point set is
\begin{equation*}
    \PS_j = \IS_j \sqcap \PS'.
\end{equation*}
Now, we have $K+1$ ``parallel'' disjoint instances of interval cover.
Formally, the $i$\th instance is $\IIS_i = (\IS_i, \PS_i)$, for
$i=0, \ldots, K$. Each instance induces a sequence of marginal
profits. Specifically, the sequence $\Pi_i$ for the $i$\th instance is
$\profitY{1}{\IIS_i} \geq \profitY{2}{\IIS_i}\geq \cdots$. By
\clmref{merge}, we need to compute the first $k$ elements in the
merged (in non-increasing order) sequence formed by these $K+1$
sequences.

Interpreting these sequences as streams, where we compute a number
only if we need it, and using heap to perform these $K+1$-way merge
results in the desired optimal solution. Specifically, for each
sequence $\Pi_i$ if the algorithm uses $t_i$ elements from it in the
current merged sequence, then the algorithm computes $t_i+1$ elements
of $\Pi_i$ using \thmref{d:p:k:cover} on $\IIS_i$. Whenever, a number
in a sequence is being consumed by the algorithm, we compute the next
number in the sequence using the ``resume'' procedure provided by
\thmref{d:p:k:cover}.

\subsubsection{Analysis}

The algorithm outputs a cover by
$(1+\eps)k$ intervals.

Let $\vopt$ be the value of the optimal cover by $k$ intervals of
$\PS' = \PS \setminus \cup \XS$. Clearly, the algorithm outputs a
cover of $\PS'$ of the same value. It follows that the value of the
solution computed by the algorithm is
\begin{equation*}
    \vopt  + \mX{ \XS} \geq \voptY{k}{\IIS}.
\end{equation*}

As for the running time, let $e_i$ be the number of elements computed
in the $i$\th sequence by the end of the execution of the
algorithm. We have that $u = \sum_i e_i = k+ K = (1+\eps)k$. Computing
an element in such a sequence takes
$T = O( \nI/K ) = O\bigl( \nI/(\eps k) \bigr)$. It follows that the
running time, ignoring presorting, is $O( T u) = O( \nI/\eps)$.

\subsubsection{The result}

\begin{theorem}
    Given an instance $(\PS, \IS)$ of $\nP$ points and $\nI$
    intervals, and parameters $k$ and $\eps$, one can a (partial)
    cover of $\PS$ by $(1+\eps)k$ intervals, in
    $O(\nP + {\nI}/{\eps} + k\log k)$ time, and
    $O(\nP + \nI/\eps)$ space. If the points and intervals are not
    presorted, the running time becomes
    $O\bigl( (\nP + \nI)\log(\nP+\nI) + \frac{\nI}{\eps} ) $.

    The computed solution covers at least as many points of $\PS$ as
    the optimal $k$-cover.
\end{theorem}

\bibliographystyle{alpha}%
\bibliography{interval_cover}

\appendix
\section{Diminishing returns for optimal covers by %
   intervals}
\apndlab{d:intervals}

In the following, we use $\optX{t}$ to denote an optimal cover with
$t$ intervals.

\subsection{Definitions and basic properties}

\subsubsection{Diminishing returns, and marginal value}

For two sets $\XS$ and $\YS$, let
$\XS \oplus \YS = (\XS \setminus \YS) \cup (\YS \setminus \XS)$ denote
their symmetric difference.

\begin{observation}
    For any sets $X, Y \subseteq \Re$, we have that
    $\mX{X \cup Y} \leq \mX{X} + \mX{Y}$. As such, we have
    $\mX{X \cup Y} - \mX{X} \leq \mX{Y}$.
\end{observation}

\begin{defn}
    \deflab{d:returns:opt}%
    For an instance $\IIS = (\PS, \IS)$ of interval cover, the
    \emphi{diminishing return} property at step $k$, for the optimal
    solution, states that $\profitX{k} \geq \profitX{k+1}$, where
    $\profitX{k} = \profitY{k}{\IIS} = \mX{\optX{k}} -
    \mX{\optX{k-1}}$.
\end{defn}

\begin{lemma}
    \lemlab{marginal}%
    For any $i$ and any $\II \in \optX{t}$, we have that
    \begin{math}
        \mvY{\II}{\optX{i}}%
        \geq %
        \profitX{i},
    \end{math}
    see \defref{marginal:value}.
\end{lemma}
\begin{proof}
    Observe that $\mX{\optX{t} - \II} \leq \mX{\optX{t-1}}$ and
    $\mX{\optX{t} - \II } = \mX{\optX{t}} -
    \mvY{\II}{\optX{t}}$. Combining we have
    $\mX{\optX{t-1}} \geq \mX{\optX{t}} - \mvY{\II}{\optX{t}}$.
\end{proof}

\subsubsection{Decomposing the optimal solutions into runs}

Let
\begin{equation*}
    \XX = \optX{k+2} \cap \optX{k}    
\end{equation*}
be the set of intervals that appear in both $\optX{k}$ and
$\optX{k+2}$.  Specifically, let $k' = k - |\XX|$, and let
\begin{equation*}
    \Li%
    =%
    (\optX{k+2} \cup \optX{k}) \setminus \XX
    =%
    \{\IJ_1, \ldots, \IJ_{2k'+2}\},
\end{equation*}
where the intervals are sorted in increasing order of their left
endpoints.  A \emphi{run} is a consecutive sequence of intervals
$\II_u, \ldots, \II_v \in \Li$, such that:
\begin{compactenumi}
    \smallskip%
    \item any two consecutive intervals in the run intersects, and
    
    \smallskip%
    \item the odd intervals belong to $\optX{k}$ and the even
    intervals belong to $\optX{k+2}$ (or vice versa).
\end{compactenumi}
\medskip%
We partition $\Li$ into the unique maximal runs from left to right,
and let $\run_1, \ldots, \run_t$ be he resulting partition of
$\Li$. The uniqueness follows as the breakpoints between runs are
pre-determined.

\begin{defn}
    The \emphw{balance} of a run $\run$ is
    $\balanceX{\run} = |\run \cap \optX{k+2}| - |\run \cap\optX{k}|$
    and it is either $-1,0,+1$. For a single interval
    $\II \in \optX{k} \cup \optX{k+2}$, we denote
    $\balanceX{ \II} = \balanceX{ \{ \II \} }$.
\end{defn}

\begin{observation}
    \obslab{two:runs}%
    Since $|\optX{k+2} | = |\optX{k}| + 2$, we have that
    $\sum_{i} \balanceX{\run_i} =2$. As such, there must be at least
    two runs in $\run_1, \ldots, \run_t$ with balance one. That is, at
    least two runs that starts and ends with an interval of
    $\optX{k+2}$.
\end{observation}

\subsection{Proving the diminishing returns property}

\begin{lemma}
    \lemlab{must:intersect}%
    If there is an interval $\II \in \optX{k+2}$ that does not
    intersect any interval of $\optX{k}$, then
    \defref{d:returns:opt} holds.
\end{lemma}

\begin{proof}
    If
    $\Delta = \mvY{\II}{\optX{k+2}} < \mX{\optX{k+2}} -
    \mX{\optX{k+1}}$ then
    \begin{equation*}
        \mX{\optX{k+2} - \II}%
        =%
        \mX{\optX{k+2}} - \Delta%
        >%
        \mX{\optX{k+2}} - \bigl( \mX{\optX{k+2}} - \mX{\optX{k+1}} \bigr)
        =%
        \mX{\optX{k+1}},
    \end{equation*}
    which is a contradiction to the optimality of $\optX{k+1}$.  Thus,
    $\Delta \geq \mX{\optX{k+2}} - \mX{\optX{k+1}}$ and we have
    \begin{equation*}
        \mX{\optX{k+1}}%
        \geq%
        \mX{\optX{k} + \II} %
        =%
        \voptX{k} + \mX{\II}%
        \geq %
        \voptX{k} + \Delta%
        \geq%
        \voptX{k} + \mX{\optX{k+2}} - \mX{\optX{k+1}},
    \end{equation*}
    which implies the claim.
\end{proof}

\begin{lemma}
    \lemlab{smaller}%
    If there is a run $\run$ of length one with $\balanceX{\run}=1$,
    then \defref{d:returns:opt} holds.
\end{lemma}

\begin{proof}
    Let $\run = \{ \II \}$, with $\II \in \optX{k+2} \setminus
    \XX$. The case not handled by \lemref{must:intersect}, is if $\II$
    intersects some interval $\IJ \in \optX{k}$.  Assume that
    $\IJ \prec \II$ (the other case is handled symmetrically), and
    $\IJ$ is the rightmost interval with this property.  If
    $\IJ \in \optX{k} \setminus \optX{k+2}$, then there must be an
    interval $\IK \in \optX{k+2} \setminus \optX{k}$ such that
    $\IJ \prec \IK \prec \II$, as otherwise $\IJ$ and $\II$ would be
    in the same run, which contradicts $\run$ being of size one.
    However, this case is impossible by \lemref{allowed:pattern}.

    Thus, $\IJ \in \XX$.  There are two possibilities:
    \begin{compactenumI}
        \smallskip%
        \item There is another interval $\IL \in \optX{k}$ that
        intersects $\II$. By the same argument $\IL \in \XX$. But
        then, the interval $\II$ has the same two neighbors in
        $\optX{k}$ and $\optX{k+2}$.  By \lemref{only:two} it has no
        other neighbors in $\optX{k}$ or $\optX{k+2}$.  Thus, removing
        $\II$ from $\optX{k+2}$, and adding it to $\optX{k}$ results
        in two covers of size $k+1$. By \lemref{marginal}, we have
        \begin{align}
          \mX{\optX{k+2}} - \mX{\optX{k+1}}
          &\leq%
            \mvY{\II}{\optX{k+2}}
            =%
            \mX{\optX{k+2} } - \mX{\optX{k+2} - \II}
            =%
            \mX{\optX{k} + \II } - \voptX{k}            
            \nonumber
          \\&%
          \leq%
          \mX{\optX{k+1} } - \voptX{k},
          \eqlab{case:a}
        \end{align}
        which establish the claim in this case.
        
        \smallskip%
        \item There is no other interval $\IL \in \optX{k}$ that
        intersects $\II$. Observe that
        \begin{equation*}
            \mvY{\II}{\optX{k+2}}
            =%
            \mX{\optX{k+2} } - \mX{\optX{k+2} - \II}
            \leq%
            \mX{\optX{k} + \II } - \voptX{k}
            \leq%
            \mX{\optX{k+1} } - \voptX{k}
        \end{equation*}
        as $\optX{k+2}$ potentially has one more interval that covers
        portions of $\II$ that is not present in $\optX{k}$.  The
        claim now follows by \lemref{marginal} and arguing as in
        \Eqref{case:a}.
    \end{compactenumI}
\end{proof}

\begin{lemma}
    \lemlab{run:no:O:k}%
    Consider a run $\run$, with $|\run|>1$, and $\balanceX{\run}=1$.
    Then any interval of $\optX{k} \setminus \run$ that intersects
    $\run$ is in $\XX$.
\end{lemma}
\begin{proof}
    Assume that $\run$ intersects some interval
    $\IJ \in \optX{k} \setminus \run$. And assume for the sake of
    contradiction that $\IJ \notin \XX$.

    If $\IJ$ was to the left of all the intervals of $\run$, then one
    assume that $\IJ$ is the rightmost such interval. Let
    $\IL \in \optX{k+2}$ be the first interval of $\run$. There must
    be another interval $\IK \in \optX{k+2}\setminus \XX$ such that
    $\IJ \prec \IK \prec \IL$, as otherwise $\IJ$ would be part of the
    run. However, this situation is impossible by
    \lemref{allowed:pattern}, as $\IJ \in \optX{k}$,
    $\IK, \IL \in \optX{k+2}$, and $\IJ \cap \IL \neq \emptyset$.

    Symmetric argument implies that $\IJ$ can not be to the right of
    all the intervals of $\run$.

    If $\IJ$ is in the middle between two consecutive intervals of the
    run, then we would have broken the run at $\IJ$, which is again
    impossible.
\end{proof}

\begin{lemma}
    \lemlab{margins}%
    Let $\run$ be a run of balance one, with $|\run|>1$, then $\run$
    can intersect at most two intervals of $\IIL, \IIR \in \XX$, where
    $\IIL \prec \IIR$. Furthermore, $\IIL$ (resp., $\IIR$) intersects
    only the first (resp. last) interval of $\run$.
\end{lemma}

\begin{proof}
    Assume for contradiction that $\IIL$ intersects some (internal)
    interval $\IJ \in \run \cap \optX{k}$, and observe that this would
    imply that $\IJ$ intersects three intervals of $\optX{k+2}$, which
    is impossible by \lemref{only:two}.  The same analysis applies if
    $\IJ \in \run \cap \optX{k+2}$ and $\IJ$ is not the first or last
    interval of $\run$.

    Thus, it must be that $\IJ$ is the last or first interval of
    $\run$, and both these intervals belong to $\optX{k+2}$.  Same
    analysis applies to $\IIR$, implying that only the first and last
    intervals in $\run$ can intersect intervals of $\run$.

    As for the second part, assume for contradiction that the (say)
    the first interval $\IJ$ in $\run$ intersects two intervals
    $\II, \II' \in \XX$, and observe that $\IJ$ then intersects three
    intervals that belongs to $\optX{k}$, which is impossible by
    \lemref{only:two}. The same argument applies to the last interval
    in the run $\run$.
\end{proof}

\subsection{The result}

\begin{theorem}
    \thmlab{diminishing:returns}%
    Let $\IIS = (\PS, \IS)$ be an instance of interval cover. We have
    that $\profitX{1} \geq \profitX{2} \geq \cdots$, where
    $\profitX{i}$ is the $i$\th marginal profit of $\IIS$ (see
    \defref{marginal:profit}).
\end{theorem}
\begin{proof}
    Consider the optimal solutions,
    $\optX{k}, \optX{k+1}, \optX{k+2}$, and their decomposition into
    runs. If there is a run of length one in this decomposition, with
    balance one, then the claim holds by \lemref{smaller}. As there
    are at least two runs in the decomposition with balance one, by
    \obsref{two:runs}, we consider such a run $\run$, with
    $\balanceX{\run} =1$, and $|\run| >1$.

    By \lemref{margins} such a run can intersect at most two intervals
    of $\XX$, and if so they intersect only the first and last
    intervals in the run.  Let $\IIL, \IIR \in \XX$ be these two
    intervals, and assume that $\IIL \prec \IIR$, with $\IIL$ (resp.,
    $\IIR$) intersect the first (resp., last) interval of $\run$.  If
    there is no interval $\IIL$ (resp., $\IIR$) with this property, we
    take $\IIL$ (resp., $\IIR$) to be the empty set.

    The argument used in the proof of \lemref{run:no:O:k} implies that
    $\run$ does not intersect any interval of
    $(\optX{k} - \IIL - \IIR) \setminus \run$.

    If there is an interval $\IJ \in \optX{k+2} \setminus \run$ that
    intersects an interval $\IK \in \run \cap \optX{k}$, then there
    are three intervals of $\optX{k+2}$ that intersect $\IK$ (two in
    the run, and $\IJ$), which is impossible by \lemref{only:two}. So
    consider the sets
    \begin{equation*}
        \LX{k} =  \cup ((\optX{k} - \IIL -\IIR) \setminus \run),\quad
        \LX{k+2} =  \cup (\optX{k+2} \setminus \run),\quad
        \RX{k} =  \cup(\optX{k} \cap \run),
        \quad\text{and}\quad%
        \RX{k+2} = \cup(\optX{k+2} \cap \run).
    \end{equation*}    
    The above argument implies that $\LX{k}$ is disjoint from both
    $\RX{k}$ and $\RX{k+2}$, and similarly, $\LX{k+2}$ and $\RX{k}$
    are disjoint.  As such
    $\mX{ \LX{k} \cup \RX{k+2} } = \mX{\LX{k} } + \mX{\RX{k+2}}$,
    $\mX{ \LX{k} \cup \RX{k}} = \mX{\LX{k}} + \mX{\RX{k}}$, and
    $\mX{\LX{k+2} \cup \RX{k}} = \mX{\LX{k+2}} + \mX{\RX{k}}$.

    Consider the two covers $\optX{k}^+ = \optX{k} \oplus \run$ and
    $\optX{k+2}^- = \optX{k+2} \oplus \run$. We have that
    $|\optX{k}^+| = |\optX{k+2}^-| = k+1$. Observe that
    $\IIL, \IIR \subseteq \LX{K+2}$, and let $\II = \IIL \cup \IIR$.
    We have that
    \begin{align*}
      \mX{\optX{k+2}} - \mX{\optX{k+1}}%
      &\leq%
        \mX{\optX{k+2}} - \mX{\optX{k+2}^-}%
      =%
      \mX{\LX{k+2}  \cup \RX{k+2} } - \mX{ \LX{k+2} \cup \RX{k}}
      \\&%
      =%
      \mX{\LX{k+2}  \cup (\RX{k+2} \setminus \II) } - \mX{\LX{k+2}}
      - \mX{ \RX{k}}
      \\&%
      \leq%
      \mX{  \RX{k+2} \setminus \II } - \mX{ \RX{k}}.
    \end{align*}
    Now, as $\LX{k} \cup \II$ is disjoint from
    $\RX{k+2} \setminus \II$ and $\RX{k}$, we have
    \begin{align*}
      \mX{\optX{k+2}} - \mX{\optX{k+1}}%
      &%
        \leq%
        \mX{ \LX{k} \cup  \II \cup (\RX{k+2} \setminus \II) } - \mX{ \LX{k}
        \cup \II \cup \RX{k}}
      \\&
      =%
      \mX{ \LX{k} \cup \II \cup  \RX{k+2} } - \mX{ \LX{k} \cup \II \cup \RX{k}}
      =%
      \mX{\optX{k}^+} - \voptX{k}%
      \\&%
      \leq%
      \mX{\optX{k+1}} - \voptX{k},
    \end{align*}
\end{proof}

\end{document}